\documentclass[runningheads]{llncs}
\usepackage[T1]{fontenc}
\usepackage{amsfonts}
\usepackage{graphicx}
\usepackage{amsmath}
\usepackage{amssymb}
\usepackage{algorithm}
\usepackage{algpseudocode}
\usepackage{booktabs}
\usepackage{multirow}
\usepackage{tikz}
\usetikzlibrary{shapes.geometric, arrows.meta}

\begin{document}
\title{TopoBudget: Persistent-Connectivity-Preserving Web Graph Sparsification for Reusable Community Analytics}
\titlerunning{Persistent-Connectivity-Preserving Graph Sparsification}

\author{Jianru Shen\orcidID{0009-0000-3546-9616}}
\authorrunning{J. Shen}
\institute{University of Montana, Missoula, MT 59812, USA \\
\email{js258133@umconnect.umt.edu}}

\maketitle              
\begin{abstract}
Web and social graphs are analyzed repeatedly for community structure, yet many of their edges are redundant for this purpose, which motivates sparsification. Existing sparsifiers preserve spectral quantities, cuts, local similarity, or a single clustering, but none preserves the thresholded connectivity structure of an edge-relevance filtration, the multiscale pattern by which groups form at high relevance and merge through weaker bridges. We study persistent-connectivity-preserving sparsification: given a graph, an edge-relevance filtration, and a proxy partition computed once during preprocessing, select a budgeted subgraph that preserves the labeled component partition at every threshold, and hence the zero-dimensional persistence diagram, while retaining community evidence for later analyses. Our method, TopoBudget, first extracts a tie-aware persistence backbone that enforces this constraint, then allocates the residual edge budget by greedily maximizing a backbone-conditioned submodular objective that rewards balanced recovery of proxy-internal degree. We prove exact preservation of the component partition at every threshold, and that the conditioned objective is monotone and submodular, so greedy attains a $(1-1/e)$ guarantee for the fixed-backbone residual problem. On held-out synthetic benchmarks and six real Web and social graphs at equal budget, TopoBudget gives the strongest community preservation among topology-preserving methods under Louvain, remains competitive under Infomap, incurs zero topology mismatch, and runs substantially faster than an effective-resistance baseline. A no-backbone ablation shows that, on the real graphs, the mandatory backbone improves average quality while providing the exact guarantee. TopoBudget thus couples exact multiscale connectivity with budgeted, reusable community preservation.
\end{abstract}

\section{Introduction}
Web platforms run on large social and information graphs: friendship and follower networks, hyperlink and co-citation graphs, and interaction graphs that underpin search, recommendation, and community discovery. These graphs are analyzed repeatedly, often by several community detectors and at several resolutions, to serve downstream applications, yet even at moderate scale they carry substantial edge redundancy. Retaining every edge inflates storage and makes each re-analysis costly, which motivates sparsification: replacing a graph by a smaller subgraph that preserves the properties of interest, so that later analyses run on the reduced graph. Existing sparsifiers each target a single objective, preserving spectral quantities, cuts, local similarity, or one predefined clustering. None of these enforces equality of the component partition across the thresholds of an edge-relevance filtration: strong intra-community ties appear at high relevance, while weaker bridges merge groups only at lower relevance. Removing the wrong bridge therefore changes this multiscale hierarchy even when the sparsified graph stays connected or a single clustering is matched.

\begin{figure}[t]
\centering
\includegraphics[width=0.86\textwidth]{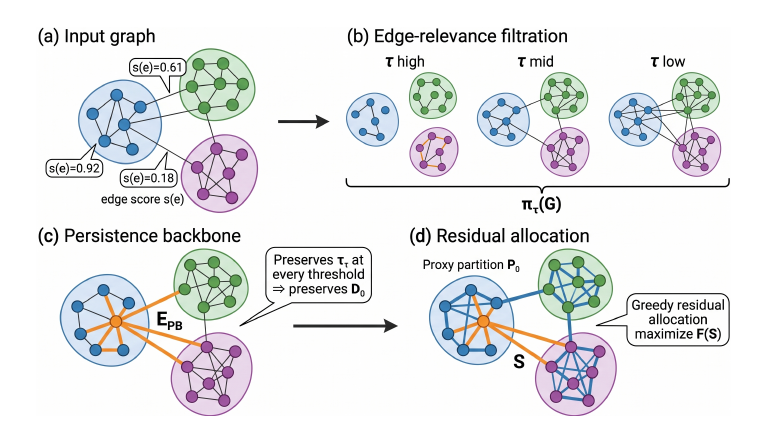}
\caption{Overview of TopoBudget. An edge-relevance filtration reveals how high-relevance groups form and merge through weaker bridges. TopoBudget first extracts a tie-aware persistence backbone $E_{\mathrm{PB}}$ that preserves $\pi_\tau(G)$ at every threshold, then allocates the residual budget $S$ by maximizing a backbone-conditioned submodular objective within the proxy partition $P_0$.}
\label{fig:overview}
\end{figure}

We formulate this missing requirement as \emph{persistent-connectivity-preserving sparsification}: given a graph, an edge-relevance filtration, and a proxy partition computed once during preprocessing, select a budgeted subgraph that preserves the labeled component partition at every threshold, and hence the zero-dimensional persistence diagram, while retaining community evidence for later analyses. This setting matches Web analytics pipelines in which one reduced graph is reused by multiple community detectors or resolutions rather than optimized for a single fixed clustering. TopoBudget separates feasibility from utility, as summarized in Figure~\ref{fig:overview}. A tie-aware persistence backbone preserves every threshold-level component partition; a backbone-conditioned submodular objective then spends the remaining budget to recover balanced proxy-internal degree. The backbone itself is a classical single-linkage certificate, but using it as a mandatory feasibility layer and optimizing the residual budget for reusable community analytics yields a sparse graph with both exact multiscale connectivity and strong community preservation.

Our contributions are fourfold: (i) a sparsification formulation that preserves the labeled component-partition filtration induced by an edge-relevance score, a requirement stronger than ordinary connectivity and stronger than preserving the unlabeled zero-dimensional persistence diagram alone; (ii) TopoBudget, a tie-aware backbone plus a backbone-conditioned submodular residual allocation objective; (iii) proofs of exact threshold-level component preservation and a $(1-1/e)$ greedy guarantee for the fixed-backbone residual problem; and (iv) an evaluation on held-out synthetic instances and six real Web and social graphs against topology-preserving and unconstrained baselines under two detectors, including a backbone ablation and preprocessing-inclusive runtime analysis.

\section{Related Work}
\label{sec:related}

\paragraph{Sparsification objectives.}
Much work reduces edge count while preserving a chosen structural quantity. Spectral sparsifiers approximate the graph Laplacian, with effective-resistance sampling bounding eigenvalues and quadratic forms~\cite{ref_spielman_teng,ref_spielman_srivastava}. Local sparsifiers keep each node's strongest incident edges, ranked by neighborhood overlap~\cite{ref_satuluri} or by neighbor degree~\cite{ref_lindner}. These methods preserve spectral quantities, cuts~\cite{ref_benczur_karger}, or local structure, but none enforces equality of the component partition at every filtration threshold, so a faithful subgraph can still merge groups at the wrong scales.

\paragraph{Community-aware sparsification.}
Closest to our motivation is community-aware sparsification, which retains clustering structure by a density or path-based criterion~\cite{ref_gionis}, and influence sparsification, which keeps a propagation backbone~\cite{ref_influence_sparsify}. These share our goal of preserving community structure~\cite{ref_fortunato}, but differ from our setting: they use an implicit rather than a fixed edge budget, are evaluated by density or path-length preservation rather than equal-budget community recovery, and give no exact multiscale-topology guarantee. Our backbone constraint is orthogonal to such residual criteria, fixing which edges are mandatory and leaving the rest of the budget free. We therefore compare against controlled equal-budget residual baselines that share the same backbone, so differences isolate how the remaining budget is allocated.

\paragraph{Topological and submodular building blocks.}
Persistent homology has been used to extract topological features of graphs for comparison or learning~\cite{ref_ph_graphs}. We instead use zero-dimensional persistence as an exact preservation constraint, relying on the classical link between single-linkage hierarchies and $D_0$ persistence~\cite{ref_singlelinkage}. The residual allocation is monotone submodular maximization with the standard greedy guarantee~\cite{ref_submodular}, a framework widely used for network selection such as influence maximization and outbreak detection~\cite{ref_influence,ref_outbreak}; stochastic variants match it in near-linear time~\cite{ref_stochastic_greedy}. Our contribution is the combination: a component-partition-filtration certificate coupled with a backbone-conditioned residual objective for reusable community analytics.

\section{Problem Formulation}
\label{sec:problem}

\paragraph{Filtration and component partitions.}
Let $G=(V,E)$ be an undirected graph with $n=|V|$ nodes and $m=|E|$ edges. We attach to every edge $e\in E$ an edge-relevance score $s(e)\in\mathbb{R}$, which measures how strongly the two endpoints are tied; in this work $s$ is the Jaccard similarity of the endpoint neighborhoods, defined in Section~\ref{sec:method}. For a threshold $\tau\in\mathbb{R}$, let
\begin{equation}
G_\tau \;=\; \bigl(V,\; \{\,e\in E : s(e)\ge \tau\,\}\bigr)
\end{equation}
be the subgraph that keeps exactly the edges with score at least $\tau$. As $\tau$ decreases from $\max_e s(e)$ to $\min_e s(e)$, edges enter in non-increasing score order and $G_\tau$ grows from the empty edge set to $G$. This nested family $\{G_\tau\}_\tau$ is the descending filtration induced by $s$. For each $\tau$, let $\pi_\tau(G)$ denote the partition of $V$ into the connected components of $G_\tau$: two nodes lie in the same block if and only if a path of edges with score at least $\tau$ joins them. As $\tau$ decreases, components merge and never split, so $\pi_\tau(G)$ coarsens, and the family $\{\pi_\tau(G)\}_\tau$ records how groups form at high relevance and merge through weaker bridges at low relevance.

\paragraph{Component-partition filtration and zero-dimensional persistence.}
We call the family $\{\pi_\tau(G)\}_\tau$ the labeled component-partition filtration of $G$: for each threshold it records not only how many components exist but which nodes each contains. Because all nodes are present throughout the edge filtration, the informative events in zero-dimensional persistent homology are the merges of components as $\tau$ decreases, and each merge is recorded by its score. Collecting these merge scores yields the zero-dimensional persistence diagram $D_0(G)$, a multiset of birth and death values. The diagram is a strictly coarser summary than the labeled filtration: preserving $\{\pi_\tau(G)\}_\tau$ implies preserving $D_0(G)$, but not conversely, since $D_0(G)$ retains only the multiset of merge scores and not which nodes belong to which component. We therefore target the stronger property, exact preservation of the labeled component-partition filtration, and obtain preservation of $D_0(G)$ as a consequence. This filtration is exactly the single-linkage merge structure of $G$~\cite{ref_ph_background,ref_singlelinkage}.

\paragraph{The general sparsification problem.}
We are given, in addition to $G$ and $s$, a proxy partition $P_0$ of $V$ obtained by running a community detector once on $G$ during preprocessing, and an edge budget $B\in\mathbb{N}$. The proxy $P_0$ is a structural reference for the objective below and is not assumed to be ground truth. Let $\Phi(E')$ be a community-preservation objective that scores an edge subset $E'\subseteq E$ by how well it recovers the community structure indicated by $P_0$; we instantiate $\Phi$ concretely in Section~\ref{sec:method}.

\begin{definition}[Persistent-connectivity-preserving sparsification]
\label{def:pcps}
Given $G=(V,E)$, edge-relevance score $s$, proxy partition $P_0$, budget $B$, and objective $\Phi$, find an edge subset $E_s\subseteq E$ with $|E_s|\le B$ such that the sparsifier $G_s=(V,E_s)$, scored by the restriction of $s$ to $E_s$, satisfies
\begin{equation}
\pi_\tau(G_s)\;=\;\pi_\tau(G)\quad\text{for every threshold }\tau,
\label{eq:exact}
\end{equation}
and, among all subsets meeting \eqref{eq:exact} and the budget, maximizes $\Phi(E_s)$.
\end{definition}

Condition~\eqref{eq:exact} is the exact multiscale requirement: at every threshold the sparsifier must induce the same labeled component partition as $G$. By the previous paragraph this implies $D_0(G_s)=D_0(G)$, and is strictly stronger than preserving the unlabeled diagram alone. It is also far stronger than merely keeping $G_s$ connected, which would constrain only the single threshold $\tau=\min_e s(e)$.

\paragraph{Score invariance, backbone, and the residual problem.}
Two observations make Definition~\ref{def:pcps} tractable. First, retained edges keep their original scores: $s$ on $E_s$ is the restriction of $s$ on $E$, never recomputed from $G_s$, so that the filtration of $G_s$ is compared to that of $G$ threshold by threshold. Second, any $E_s$ satisfying \eqref{eq:exact} must contain enough edges to reproduce every merge of $G$; the minimal such set is a maximum-score spanning forest of $G$, which we call the persistence backbone $E_{\mathrm{PB}}$ and construct in Section~\ref{sec:method}. Its size is $|E_{\mathrm{PB}}|=n-c$, where $c$ is the number of connected components of $G$; we preprocess each graph to its largest connected component, so $c=1$ and $E_{\mathrm{PB}}$ is a spanning tree with $n-1$ edges. Feasibility therefore requires $B\ge|E_{\mathrm{PB}}|=n-1$.

Writing $E_s=E_{\mathrm{PB}}\cup S$ with $S\subseteq E\setminus E_{\mathrm{PB}}$, the budget splits as $|E_s|=|E_{\mathrm{PB}}|+|S|$. Fixing the backbone reduces the search to the residual set $S$, which is the problem TopoBudget solves.

\begin{definition}[Backbone-conditioned residual problem]
\label{def:residual}
Given a fixed persistence backbone $E_{\mathrm{PB}}$ satisfying \eqref{eq:exact} and a residual budget $k=B-|E_{\mathrm{PB}}|$, find $S\subseteq E\setminus E_{\mathrm{PB}}$ with $|S|\le k$ that maximizes the residual objective
\begin{equation}
F(S)\;=\;\Phi(E_{\mathrm{PB}}\cup S)-\Phi(E_{\mathrm{PB}}).
\label{eq:residual_general}
\end{equation}
\end{definition}

Every $E_{\mathrm{PB}}\cup S$ is feasible for Definition~\ref{def:pcps} by construction, so any solution of Definition~\ref{def:residual} is feasible for the general problem. When the maximum-score spanning forest is not unique, different valid backbones exist; our guarantees are stated for the fixed backbone that TopoBudget constructs, not for a search over all backbones. This conditioning is what makes the residual allocation a submodular optimization with a provable guarantee, as shown in Section~\ref{sec:theory}.

\section{Method}
\label{sec:method}

TopoBudget realizes the residual problem of Definition~\ref{def:residual} in three stages: it scores edges by neighborhood similarity, extracts a tie-aware persistence backbone $E_{\mathrm{PB}}$ that already satisfies the exact-topology constraint~\eqref{eq:exact}, and then spends the residual budget on a backbone-conditioned submodular objective driven by the proxy partition $P_0$.

\paragraph{Edge-relevance score.}
For an unweighted graph we score each edge by the Jaccard similarity of its endpoint neighborhoods. Let $N(u)$ denote the set of neighbors of node $u$. For an edge $e=(u,v)$,
\begin{equation}
s(u,v)\;=\;\frac{|N(u)\cap N(v)|}{|N(u)\cup N(v)|}.
\label{eq:jaccard}
\end{equation}
A high score means the endpoints share many common neighbors, which is typical of an edge inside a community, while a low score is typical of a bridge whose endpoints have few neighbors in common. Ranking edges by \eqref{eq:jaccard} therefore makes the descending filtration of Section~\ref{sec:problem} reveal communities at high relevance and bridges at low relevance.

\paragraph{Stage I: tie-aware persistence backbone.}
The backbone $E_{\mathrm{PB}}$ is the minimal edge set that reproduces every merge of the filtration. We process edges in non-increasing score order and maintain the connected components with a union-find structure: an edge is added to $E_{\mathrm{PB}}$ if and only if its endpoints lie in different components when it is processed, that is, if it triggers a merge. This yields a maximum-score spanning forest, so $|E_{\mathrm{PB}}|=n-c$ with $c=1$ after restriction to the largest connected component. Ties need care: when several edges share a score $\tau$, the component partition at $\tau$ must equal the one obtained once all of them are present. We therefore process an equal-score block together, and within the block add exactly the edges that merge components that are still distinct, which is a spanning forest of the block's contribution. This guarantees $\pi_\tau(G_{\mathrm{PB}})=\pi_\tau(G)$ at the tie value as well, and Section~\ref{sec:theory} proves the property for every threshold.

\paragraph{Stage II: proxy partition.}
We obtain the proxy partition $P_0$ by running a single community detection on the full graph $G$ during preprocessing; we use Leiden~\cite{ref_leiden}. The proxy tells the residual objective which edges are internal to a candidate community, and its cost is counted in preprocessing. It is not treated as ground truth. We note that Leiden and the primary evaluation detector Louvain both optimize modularity~\cite{ref_modularity}, so to avoid crediting the objective for a shared inductive bias we additionally evaluate with Infomap, which uses a different, flow-based criterion; the primary comparisons in Section~\ref{sec:experiments} report both.

\paragraph{Stage III: backbone-conditioned degree-coverage objective.}
The residual budget is allocated to preserve community structure through a concave coverage of internal degree, which instantiates the abstract objective $\Phi$ of Definition~\ref{def:pcps}. For an edge set $T\subseteq E$ and a node $u$, let $d_T^{\mathrm{in}}(u)$ be the number of neighbors of $u$ in $T$ that lie in the same block of $P_0$ as $u$, the internal degree of $u$ in $T$. Let $V_{\mathrm{int}}=\{u\in V: d_G^{\mathrm{in}}(u)>0\}$ be the nodes with at least one internal edge in $G$. For each $u\in V_{\mathrm{int}}$ define a normalization scale
\begin{equation}
\kappa_u\;=\;\max\bigl(1,\,\lceil \rho\, d_G^{\mathrm{in}}(u)\rceil\bigr),
\label{eq:kappa}
\end{equation}
with $\rho>0$ fixed to $\rho=2$. The coverage objective is
\begin{equation}
\Phi(T)\;=\;U(T)\;=\;\frac{1}{|V_{\mathrm{int}}|}\sum_{u\in V_{\mathrm{int}}}\sqrt{\min\!\left(1,\;\frac{d_T^{\mathrm{in}}(u)}{\kappa_u}\right)}.
\label{eq:objective}
\end{equation}
The square root makes the first internal edges of a node the most valuable and the marginal value decrease as its internal neighborhood fills, which spreads the budget across nodes rather than concentrating it on a few high-degree hubs. With $\rho=2$ the ratio $d_T^{\mathrm{in}}(u)/\kappa_u$ stays below about one half, so $\kappa_u$ acts as a per-node normalization that keeps the terms comparable rather than a saturation level a node can reach; the inner $\min$ caps each term at one and keeps $U$ bounded.

Conditioning on the backbone yields the residual objective of Definition~\ref{def:residual}. For a candidate set $S\subseteq E\setminus E_{\mathrm{PB}}$,
\begin{equation}
F(S)\;=\;U(E_{\mathrm{PB}}\cup S)-U(E_{\mathrm{PB}}),
\label{eq:residual}
\end{equation}
so that $F(\varnothing)=0$ and $F$ credits only the internal degree the backbone does not already supply. Section~\ref{sec:theory} shows that $F$ is monotone and submodular.

\paragraph{Budget and algorithm.}
Given a residual fraction $\alpha\in(0,1)$, the total edge budget is $B=|E_{\mathrm{PB}}|+k$ with $k=\lfloor \alpha\,(m-|E_{\mathrm{PB}}|)\rfloor$, so the sparsifier keeps the backbone and $k$ further edges. We select the residual set $S$ by greedy maximization of $F$: starting from $S=\varnothing$, repeatedly add the candidate edge of largest marginal gain $F(S\cup\{e\})-F(S)$ until $|S|=k$, with lazy evaluation of marginal gains. Algorithm~\ref{alg:topobudget} summarizes the procedure. For graphs large enough that exact greedy becomes costly, the same objective admits a stochastic greedy variant that samples a small candidate pool each round and retains a $(1-1/e-\epsilon)$ guarantee~\cite{ref_stochastic_greedy}; we use exact greedy throughout, as it is fast at the scale of our graphs, and leave a large-scale study of the stochastic variant to future work.

\begin{algorithm}
\caption{TopoBudget}
\label{alg:topobudget}
\begin{algorithmic}[1]
\State \textbf{Input:} graph $G=(V,E)$, score $s$, proxy $P_0$, residual fraction $\alpha$
\State compute $s(e)$ for all $e\in E$ \Comment{Eq.~\eqref{eq:jaccard}}
\State $E_{\mathrm{PB}}\gets$ maximum-score spanning forest by tie-aware union-find
\State $k\gets \lfloor \alpha\,(m-|E_{\mathrm{PB}}|)\rfloor$;\quad $S\gets\varnothing$
\While{$|S|<k$}
  \State $e^\ast\gets \arg\max_{e\in (E\setminus E_{\mathrm{PB}})\setminus S} F(S\cup\{e\})-F(S)$
  \State $S\gets S\cup\{e^\ast\}$
\EndWhile
\State \textbf{return} $E_s=E_{\mathrm{PB}}\cup S$
\end{algorithmic}
\end{algorithm}

\section{Theoretical Guarantees}
\label{sec:theory}

We prove the two properties used in Definitions~\ref{def:pcps} and~\ref{def:residual}: every TopoBudget output preserves the multiscale connectivity exactly, and the residual objective is monotone and submodular, so greedy allocation is near optimal. Throughout, $G_\tau$, $\pi_\tau$, and $D_0$ are as in Section~\ref{sec:problem}, and for an edge set $A\subseteq E$ we write $A_\tau=\{e\in A: s(e)\ge\tau\}$ for its edges of score at least $\tau$.

\subsection{Exact Multiscale Preservation}

We first record the property of the backbone that drives the argument.

\begin{lemma}
\label{lem:forest}
Let $E_{\mathrm{PB}}$ be a maximum-score spanning forest of $G$ built by the tie-aware union-find of Section~\ref{sec:method}. Then for every threshold $\tau$, the graph $(V,(E_{\mathrm{PB}})_\tau)$ has the same connected components as $G_\tau$, that is $\pi_\tau(V,(E_{\mathrm{PB}})_\tau)=\pi_\tau(G)$.
\end{lemma}

\begin{proof}
Fix $\tau$. Since $(E_{\mathrm{PB}})_\tau\subseteq E_\tau$, any two nodes connected in $(V,(E_{\mathrm{PB}})_\tau)$ are connected in $G_\tau$, so the components of $(E_{\mathrm{PB}})_\tau$ refine those of $G_\tau$. For the converse, let $u,v$ lie in the same component of $G_\tau$, joined by a path $P$ using only edges of score at least $\tau$. Consider any edge $e=(a,b)$ on $P$. The union-find processes edges in non-increasing score order; at the moment $e$ is processed, either $a$ and $b$ already lie in the same component, in which case they are joined by previously added backbone edges of score at least $s(e)\ge\tau$, or $e$ is added to $E_{\mathrm{PB}}$. In both cases $a$ and $b$ are connected in $(V,(E_{\mathrm{PB}})_\tau)$. Chaining over the edges of $P$ connects $u$ and $v$ in $(V,(E_{\mathrm{PB}})_\tau)$. Hence the components of $G_\tau$ refine those of $(E_{\mathrm{PB}})_\tau$, and the two partitions coincide. The argument uses only $s(e)\ge\tau$, so it holds at tie values, where all equal-score edges are processed together as in Section~\ref{sec:method}. \qed
\end{proof}

\begin{theorem}
\label{thm:exact}
Let $E_s=E_{\mathrm{PB}}\cup S$ with $S\subseteq E\setminus E_{\mathrm{PB}}$ be any TopoBudget output, scored by the restriction of $s$ to $E_s$. Then $\pi_\tau(G_s)=\pi_\tau(G)$ for every threshold $\tau$, and consequently $D_0(G_s)=D_0(G)$.
\end{theorem}

\begin{proof}
Fix $\tau$. Because $E_{\mathrm{PB}}\subseteq E_s\subseteq E$, taking score-at-least-$\tau$ subsets preserves the inclusions $(E_{\mathrm{PB}})_\tau\subseteq (E_s)_\tau\subseteq E_\tau$. The right inclusion makes the components of $(E_s)_\tau$ refine those of $G_\tau$, since adding edges only merges components. The left inclusion with Lemma~\ref{lem:forest} gives the reverse: any two nodes in the same component of $G_\tau$ are already connected in $(V,(E_{\mathrm{PB}})_\tau)$, hence in $(V,(E_s)_\tau)$. The two partitions therefore coincide, $\pi_\tau(G_s)=\pi_\tau(G)$.

Since all nodes persist through the filtration, the zero-dimensional persistence diagram is determined by the merges recorded in $\{\pi_\tau\}_\tau$: each death is a threshold at which two blocks of $\{\pi_\tau\}_\tau$ join, and these merge scores are read directly off the family. As $\pi_\tau(G_s)=\pi_\tau(G)$ for all $\tau$, the merge scores agree, so $D_0(G_s)=D_0(G)$. Preserving the labeled family $\{\pi_\tau\}_\tau$ is strictly stronger than preserving $D_0$ alone, and TopoBudget attains the stronger property. \qed
\end{proof}

Theorem~\ref{thm:exact} holds for every feasible budget, since it constrains only the inclusion $E_{\mathrm{PB}}\subseteq E_s$ and not the size of $S$. Section~\ref{sec:experiments} verifies it empirically: every backbone-constrained method attains zero component-partition mismatch across all sampled thresholds, whereas unconstrained sparsifiers break the partition at a large fraction of them.

\subsection{Monotonicity, Submodularity, and Greedy Guarantee}

Recall from Section~\ref{sec:method} that $U(T)=\frac{1}{|V_{\mathrm{int}}|}\sum_{u\in V_{\mathrm{int}}} g\!\left(d_T^{\mathrm{in}}(u)/\kappa_u\right)$ with $g(z)=\sqrt{\min(1,z)}$, and that the residual objective is $F(S)=U(E_{\mathrm{PB}}\cup S)-U(E_{\mathrm{PB}})$.

\begin{theorem}
\label{thm:submodular}
For a fixed backbone $E_{\mathrm{PB}}$, the residual objective $F$ is normalized, monotone non-decreasing, and submodular on subsets of $E\setminus E_{\mathrm{PB}}$. Consequently, greedy maximization of $F$ under the cardinality budget $k$ returns a set $S$ with $F(S)\ge(1-1/e)\,F(S^\ast)$, where $S^\ast$ is an optimal solution of the backbone-conditioned residual problem of Definition~\ref{def:residual}.
\end{theorem}

\begin{proof}
Fix a node $u\in V_{\mathrm{int}}$. For an edge set $T$, the internal degree $d_T^{\mathrm{in}}(u)=\sum_{e}\mathbf{1}[e\in T]$, where the sum runs over edges incident to $u$ whose other endpoint lies in the same block of $P_0$ as $u$. As a sum of indicators over a fixed edge set, $d_T^{\mathrm{in}}(u)$ is a non-negative modular function of $T$, rising by one when an internal edge of $u$ is added and unchanged otherwise; dividing by the constant $\kappa_u$ preserves modularity.

The map $g(z)=\sqrt{\min(1,z)}$ is non-decreasing and concave on $[0,\infty)$. Composing a non-decreasing concave function with a non-negative modular function yields a monotone submodular function: monotonicity is immediate since both maps are non-decreasing, and submodularity holds because the discrete marginal $g(z+\delta)-g(z)$ is non-increasing in $z$ for fixed $\delta\ge 0$ by concavity, so the gain of adding an internal edge to $u$ shrinks as $u$ accumulates internal edges. Each term $g(d_T^{\mathrm{in}}(u)/\kappa_u)$ is therefore monotone submodular, and $U$, a non-negative linear combination of these terms with weights $1/|V_{\mathrm{int}}|>0$, is monotone submodular as well.

Conditioning on a fixed set preserves both properties: $F(S)=U(E_{\mathrm{PB}}\cup S)-U(E_{\mathrm{PB}})$ is the marginal of a monotone submodular function over the base $E_{\mathrm{PB}}$, hence monotone and submodular on $E\setminus E_{\mathrm{PB}}$, with $F(\varnothing)=0$. The $(1-1/e)$ bound for greedy maximization of a monotone submodular function under a cardinality constraint then applies~\cite{ref_submodular}. \qed
\end{proof}

Theorem~\ref{thm:submodular} justifies the greedy allocation in Algorithm~\ref{alg:topobudget}: for the backbone TopoBudget constructs, the chosen residual edges are within a $(1-1/e)$ factor of the best possible internal-degree coverage at the given budget. We do not claim optimality over all valid backbones; the guarantee is for the residual allocation conditioned on the constructed backbone. The stochastic greedy variant retains a $(1-1/e-\epsilon)$ bound in expectation~\cite{ref_stochastic_greedy}, at lower cost on large graphs.

\section{Experiments}
\label{sec:experiments}

\begin{figure}[t]
\centering
\includegraphics[width=0.86\textwidth]{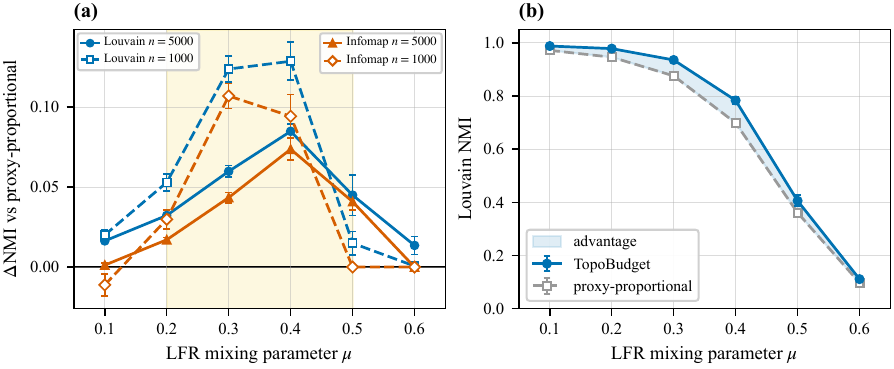}
\caption{Held-out conditional advantage on unseen LFR instances. (a) NMI gain of TopoBudget over the proxy-proportional control as a function of the mixing parameter $\mu$, for Louvain (blue) and Infomap (orange); solid lines are $n=5000$ and dashed lines are $n=1000$. The shaded band marks the intermediate-mixing region where the gain is largest. (b) Absolute Louvain NMI of TopoBudget and the proxy-proportional control at $n=5000$, with the shaded area showing TopoBudget's advantage. The advantage follows an inverted-U shape, largest at intermediate mixing.}
\label{fig:conditional}
\end{figure}

\paragraph{Setup.}
We evaluate on six real Web and social graphs and on synthetic LFR benchmarks; dataset statistics are in Table~\ref{tab:datasets}. Community structure is recovered with two detectors that are independent of the proxy, Louvain~\cite{ref_louvain} as the primary detector and Infomap~\cite{ref_infomap} as a second, and quality is measured by the normalized mutual information (NMI)~\cite{ref_nmi_danon,ref_nmi} between the partition of the sparsifier and the partition of the original graph at equal budget. We report results at residual fractions $\alpha\in\{0.1,0.3,0.5\}$. The objective was frozen before these experiments, and all synthetic evaluation uses held-out LFR instances generated with seeds disjoint from development.

\paragraph{Datasets and preprocessing.}
The six real graphs are drawn from the SNAP collection~\cite{ref_snap,ref_egofacebook} and the Wikipedia and LastFM datasets of~\cite{ref_musae}. They span densities from sparse collaboration (ca-GrQc, average degree 6.5) to dense Web graphs (squirrel, average degree 76.3) and include social, Web, and collaboration networks. Each graph is reduced to a simple undirected graph: directed inputs are symmetrized with the OR rule, self-loops and multi-edges are removed, and we restrict to the largest connected component, so $c=1$ and the backbone is a spanning tree with $n-1$ edges. All methods are compared under the same unweighted equal-budget model: every sparsifier returns an edge subset of the same size, and community detection runs on the resulting unweighted subgraph. In particular, effective resistance is computed on the full unweighted graph and used only to rank and select edges, so no method benefits from retaining edge weights; the backbone-constrained baselines share the identical backbone and differ only in how the residual budget is spent. For each graph the full-graph partition is the reference. The synthetic family uses the LFR benchmark~\cite{ref_lfr}, which plants communities with a controllable mixing parameter $\mu$, the fraction of each node's edges that cross community boundaries; larger $\mu$ makes communities harder to recover. We generate LFR graphs at $n\in\{1000,5000\}$ for $\mu\in\{0.1,\dots,0.6\}$ with ten seeds per configuration, disjoint from those used during development.

\begin{table}[t]
\caption{Dataset statistics. Six real Web and social graphs and the synthetic LFR family. Average degree is $2m/n$. The LFR family is generated at two sizes across a range of mixing parameters, with ten held-out seeds per configuration.}
\label{tab:datasets}
\centering
\begin{tabular}{lrrrl}
\hline
Graph & $n$ & $m$ & Avg.\ deg. & Type\\
\hline
ego-Facebook & 4039 & 88234 & 43.7 & social\\
Wiki-Vote & 7066 & 100736 & 28.5 & social\\
ca-GrQc & 4158 & 13422 & 6.5 & collaboration\\
chameleon & 2277 & 31371 & 27.6 & Web\\
squirrel & 5201 & 198353 & 76.3 & Web\\
LastFM Asia & 7624 & 27806 & 7.3 & social\\
\hline
LFR & 1000, 5000 & varies & 14.5 & synthetic\\
\hline
\end{tabular}
\end{table}

\begin{table}[t]
\caption{Community preservation on held-out LFR instances at $n=5000$, by mixing parameter $\mu$, averaged over ten seeds and budgets $\alpha\in\{0.1,0.3,0.5\}$. NMI is measured against the planted partition. All four methods retain the backbone; they differ only in residual allocation. The best value in each column is in bold. TopoBudget is highest at every mixing level under both detectors.}
\label{tab:lfr}
\centering
\begin{tabular}{llccc}
\hline
Detector & Method & $\mu=0.2$ & $\mu=0.3$ & $\mu=0.4$\\
\hline
\multirow{4}{*}{Louvain}
 & \textbf{TopoBudget (Ours)} & \textbf{0.979} & \textbf{0.936} & \textbf{0.784}\\
 & Proxy-Prop. & 0.946 & 0.876 & 0.699\\
 & Local Degree & 0.887 & 0.798 & 0.607\\
 & Random & 0.857 & 0.712 & 0.463\\
\hline
\multirow{4}{*}{Infomap}
 & \textbf{TopoBudget (Ours)} & \textbf{0.962} & \textbf{0.949} & \textbf{0.878}\\
 & Proxy-Prop. & 0.945 & 0.906 & 0.804\\
 & Local Degree & 0.902 & 0.857 & 0.754\\
 & Random & 0.884 & 0.799 & 0.663\\
\hline
\end{tabular}
\end{table}

\begin{table}[t]
\caption{Equal-budget comparison on the six real graphs, averaged over graphs and budgets $\alpha\in\{0.1,0.3,0.5\}$. NMI is measured against the original-graph partition for each detector. Mismatch rate is the fraction of filtration thresholds at which the component partition differs from that of the original graph (0 = exact preservation). Runtime is the mean sparsification time. The best value in each column among topology-preserving methods is in bold; methods are grouped into backbone-constrained (top), unconstrained (middle), and the no-backbone ablation of our objective (bottom). Our method is highlighted.}
\label{tab:main}
\centering
\begin{tabular}{lcccr}
\hline
Method & Louvain NMI & Infomap NMI & Mismatch rate & Runtime (s)\\
\hline
\textbf{TopoBudget (Ours)} & \textbf{0.739} & 0.784 & \textbf{0.000} & 3.39\\
PB + Eff. Resistance & 0.703 & \textbf{0.800} & \textbf{0.000} & 24.94\\
PB + Proxy-Prop. & 0.699 & 0.759 & \textbf{0.000} & \textbf{0.04}\\
PB + Local Similarity & 0.662 & 0.754 & \textbf{0.000} & 0.19\\
PB + Local Degree & 0.623 & 0.724 & \textbf{0.000} & 0.01\\
\hline
Eff. Resistance & 0.701 & 0.792 & 0.983 & 24.82\\
Local Similarity & 0.606 & 0.738 & 0.615 & 0.20\\
Local Degree & 0.520 & 0.630 & 0.983 & 0.02\\
\hline
TopoBudget, no backbone & 0.726 & 0.764 & 0.996 & 0.18\\
\hline
\end{tabular}
\end{table}

\begin{figure}[t]
\centering
\includegraphics[width=0.86\textwidth]{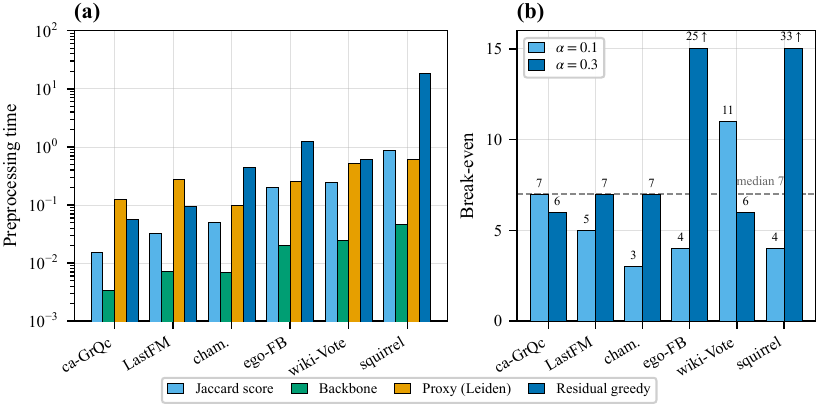}
\caption{Preprocessing cost and workload break-even on the real graphs. (a) Preprocessing time in seconds (log scale), decomposed into Jaccard scoring, backbone construction, proxy detection, and residual greedy, at $\alpha=0.3$; graphs ordered by edge count. The backbone is negligible. (b) Workload break-even, the number of Louvain analyses on the sparsified graph that amortize preprocessing, at $\alpha=0.1$ and $\alpha=0.3$; the y-axis is capped at 15 with larger values annotated. A handful of reuses suffice.}
\label{fig:runtime}
\end{figure}

\paragraph{Held-out conditional advantage.}
Figure~\ref{fig:conditional} reports the NMI gain of TopoBudget over the proxy-proportional control on held-out LFR instances, as a function of the mixing parameter $\mu$. The advantage follows an inverted-U shape: it is small when communities are nearly separable or nearly dissolved, and largest at intermediate mixing where the choice of retained internal edges matters most. The pattern holds for both detectors and across graph sizes, confirming that the frozen objective carries a real and reproducible benefit on unseen instances rather than an artifact of tuning. Table~\ref{tab:lfr} gives the absolute NMI at $n=5000$, where TopoBudget is highest at every mixing level under both detectors.

\paragraph{Equal-budget baseline comparison.}
Table~\ref{tab:main} compares TopoBudget against backbone-constrained and unconstrained sparsifiers on the six real graphs at equal budget. Three findings stand out. First, on the primary detector TopoBudget attains the highest NMI on average among all methods that preserve exact multiscale topology, ahead of the effective-resistance, proxy-proportional, local-similarity, and local-degree baselines. Second, on Infomap TopoBudget is competitive with the strongest topology-preserving baseline, effective resistance, which edges ahead on this detector; both, however, preserve topology exactly, with a component-partition mismatch rate of zero, whereas every unconstrained sparsifier breaks the partition at most thresholds. Third, TopoBudget reaches this quality far faster than effective resistance: on most graphs it is one to two orders of magnitude faster, and only on the densest graph, squirrel, are the two comparable, which motivates the stochastic variant for dense inputs. Exact preservation is confirmed directly on every graph: each backbone-constrained method attains zero component-partition mismatch at all sampled thresholds, from 74 to 100 per graph, while the unconstrained spectral and degree baselines break the partition at most of them.

\paragraph{Role of the backbone.}
The no-backbone row applies the same objective without the mandatory backbone. On the six real graphs the backbone is beneficial: TopoBudget attains higher NMI than the no-backbone variant, by $0.013$ under Louvain and $0.020$ under Infomap, while additionally guaranteeing zero topology mismatch. The no-backbone variant, despite dropping the guarantee, does not gain quality in aggregate on the real graphs, and it breaks the multiscale structure at almost every threshold. The backbone therefore secures exact multiscale preservation at no cost to real-graph community quality.

\paragraph{Objective-component ablation.}
Table~\ref{tab:objective-ablation} isolates the residual objective design while keeping the backbone, budget, and evaluation fixed. The pair term is an alternative proxy-community-pair coverage term evaluated only as a design variant and is not used in the final method. The degree-only objective achieves the highest mean Louvain NMI on both held-out LFR instances and real graphs, so TopoBudget uses the degree-only form.

\begin{table}[t]
\centering
\caption{Objective-component ablation. Mean Louvain NMI, averaged over budgets; higher is better. All variants use the same persistence backbone and differ only in residual allocation.}
\label{tab:objective-ablation}
\small
\setlength{\tabcolsep}{6pt}
\begin{tabular}{lcc}
\hline
Objective variant & LFR & Real graphs \\
\hline
Proxy-proportional & 0.837 & 0.699 \\
Pair-only ($\lambda=0$) & 0.821 & 0.675 \\
Full ($\lambda=0.5$) & 0.864 & 0.727 \\
Degree-only ($\lambda=1$, ours) & \textbf{0.904} & \textbf{0.739} \\
\hline
\end{tabular}
\end{table}

\paragraph{Runtime and workload break-even.}
Figure~\ref{fig:runtime} decomposes preprocessing time and reports the workload break-even. The backbone construction is negligible, under $0.3\%$ of preprocessing on every graph, and the cost is dominated by proxy detection and residual selection. Because TopoBudget is built once and reused, a small number of analyses amortizes the preprocessing: the break-even is a median of seven Louvain analyses at $\alpha=0.3$ and fewer at higher sparsity. On the densest graph the exact residual greedy is the main expense, which again points to the stochastic variant as the scalable path.

\section{Conclusion and Limitations}
\label{sec:conclusion}

We introduced persistent-connectivity-preserving sparsification, which reduces a graph to a budgeted subgraph while exactly preserving the labeled component partition at every filtration threshold. TopoBudget enforces this constraint with a tie-aware persistence backbone and then allocates the residual edge budget by a backbone-conditioned submodular objective. The result is a reusable sparse representation that couples exact multiscale connectivity with budgeted community preservation.

The experiments show where this constraint is useful. TopoBudget is the strongest topology-preserving method on average under Louvain and remains competitive with effective resistance under Infomap, while every unconstrained sparsifier breaks the component partition at many thresholds. The no-backbone ablation shows that the backbone supplies a categorical guarantee and, on the real graphs, improves rather than hurts average NMI. The method is therefore most appropriate when one sparsified graph must support repeated or alternative analyses rather than a single detector optimized in isolation.

Three limitations remain. First, exact greedy is the main cost on the densest graph; the stochastic greedy variant has the same approximation guarantee up to an additive $\epsilon$, but its large-scale empirical validation is future work. Second, the residual objective depends on the proxy partition, so weak proxies provide weaker guidance. Third, the Jaccard filtration is not equally informative on all graphs, suggesting future work on richer edge-relevance filtrations. These limitations affect empirical quality and scalability, but not the exact preservation guarantee, which holds for every graph.

\begin{credits}
\subsubsection{\discintname}
The authors have no competing interests to declare that are relevant to the content of this article.
\end{credits}

%
%
%
\bibliographystyle{splncs04}
\bibliography{references}

\end{document}